\documentclass[preprint,12pt,authoryear]{elsarticle}

\usepackage{amssymb}
\usepackage{amsmath}
\usepackage{amsthm}
\usepackage{graphicx}
\usepackage{booktabs}
\usepackage[font=small,labelfont=bf]{caption}
\usepackage{url}

\newtheorem{theorem}{Theorem}
\newtheorem{proposition}{Proposition}

\theoremstyle{remark}
\newtheorem{remark}{Remark}

\journal{Spatial Statistics}

\begin{document}

\begin{frontmatter}

\title{Separating Spatial and Clinical Risk with Node-Splitting SVM Survival Trees}

\author[bsu]{Drew Lazar\corref{cor1}}
\ead{dmlazar@bsu.edu}
\cortext[cor1]{Corresponding author.}

\author[uofl]{Aye Aye Maung}
\ead{ayeaye.maung@louisville.edu}

\address[bsu]{Department of Mathematical Sciences, Ball State University,
Robert Bell Building, Room 465, Muncie, IN 47306, USA}

\address[uofl]{Department of Bioinformatics and Biostatistics, University of
Louisville, 485 E. Gray St., Louisville, KY 40202, USA}

\begin{abstract}
Recovering geographic variation in survival requires separating spatial risk from
patients' clinical characteristics, a problem complicated by prognostic covariates
that are themselves spatially structured. We develop a
nonparametric two-stage method for this separation. A clinical survival tree fit
to the covariates alone supplies leaf Nelson-Aalen cumulative hazard residuals,
transferring the censored survival structure to a clinically adjusted scale
without imposing a functional form on the clinical hazard, and a second tree fit
to these residuals on the coordinates recovers the spatial structure. Both stages are
kernel dipole-splitting survival trees, so the resulting spatial risk map is
piecewise constant, with sharp, possibly curved boundaries. We establish when the residuals recover the
spatial signal: under a multiplicative frailty and an exogeneity
condition, they are free of the clinical covariates given location and
stochastically ordered by the frailty. An expansion of the frailty Laplace exponent quantifies their approximate
exponentiality and identifies the leading remainder. These results assume consistency of the first stage rather
than a model class, so the construction extends to other survival estimators. On the
LeukSurv leukemia data the method agrees with a Bayesian Gaussian random field
frailty about where risk is elevated while resolving sharp adjacencies the smooth
surface averages away, and an unadjusted spatial analysis misattributes
clinical variation to location. Simulations with known zones,
including a sweep through graded violations of exogeneity, locate the point at
which the two contributions cease to be separately identifiable, with the smooth
benchmark degrading in parallel as that point is approached.
\end{abstract}


\end{frontmatter}

\section{Introduction}
\label{sec:intro}

Survival from serious disease varies geographically. Cancer registries in many
countries record regional differences that persist after adjustment for patients'
clinical characteristics, and locating them is a standard object of epidemiological
interest. A region of elevated risk unexplained by its case mix points to something
about the place itself \citep{clayton1993, woods2006}. The statistical problem is to
separate the geographic component of risk from the clinical component, using data in
which survival times are right-censored and in which location and clinical
characteristics are themselves related.

The established approach adds a spatially structured frailty to the linear predictor
of a proportional-hazards model. \citet{henderson2002} introduced this framework in
the setting that motivates the present paper, modeling survival among acute myeloid
leukemia patients in northwest England with a multiplicative gamma frailty carrying a
Mat\'ern spatial correlation, and related developments followed for point-referenced
data \citep{li2002} and for areal data through conditionally autoregressive priors
\citep{banerjee2003}. The modern implementation of this family is the
\texttt{spBayesSurv} package of \citet{zhou2020spbayessurv}, which places a flexible
transformed Bernstein polynomial prior on the baseline and admits proportional-odds
and accelerated-failure-time specifications alongside proportional hazards. Across
all of these models the spatial effect is a smooth, distance-decaying surface, so the
estimated risk varies continuously over the study region.

That smoothness is a modeling assumption rather than an empirical finding, and there
are reasons to doubt it where spatial survival analysis is applied. Health care is
delivered through discrete administrative units with distinct catchment areas and
protocols \citep{goddard2001}, and environmental exposures change as abruptly, across
a river or a former industrial boundary \citep{bakka2019}. Where the underlying risk
changes sharply, a correlation-based frailty cannot represent the change. The
Mat\'ern structure forces neighboring locations toward common values, so a genuine
discontinuity is returned as a gradient and two adjacent regions of opposite risk are
averaged into an intermediate value. The proportional-hazards assumption on the
clinical part of the model is a second restriction. It leaves the shared baseline
hazard unconstrained but requires each covariate's hazard ratio to be constant over
follow-up, and it fails most visibly when a covariate effect reverses direction, as
an age effect can when early mortality and long-term prognosis are governed by
different mechanisms.

We take a different approach to the spatial component, using survival trees. A tree
partitions the covariate space into regions and fits a separate nonparametric
survival estimate in each, producing a piecewise-constant risk surface with sharp
boundaries, which is precisely the representation a smooth frailty cannot supply. We
build on the kernel dipole-splitting survival tree of \citet{maung2025}, in which
splits are chosen by a margin-based criterion constructed from pairs of observations,
or \emph{dipoles}, graded by their survival-time differences and censoring pattern.
Because the criterion is a regularized hinge-loss objective, its dual depends on the
covariates only through inner products, and kernelization yields curved separating
surfaces while retaining the discrete leaf structure of a tree.

Fitting such a tree to raw survival times, however, does not answer the question of
interest. Clinical covariates that are prognostic and also spatially organized
(area-level deprivation being the canonical example, \citealt{townsend1988,
carstairs1989}) induce geographic variation in survival that is clinical in
origin, and a tree given only coordinates will attribute it to location
\citep{reich2006}. We therefore
introduce a two-stage procedure. A first tree is fit to the clinical covariates alone;
from its leaves we compute Nelson-Aalen cumulative hazard residuals, which transfer
the censored survival structure onto a clinically adjusted scale without imposing a
functional form on the clinical hazard; a second tree is then fit to these residuals
using only the coordinates. The result is a discrete spatial risk map read as the
spatial effect net of clinical adjustment. The procedure shares the two-stage logic of
regression kriging \citep{hengl2007}, but preserves censoring and partitions where
kriging smooths.

The contributions of the paper are as follows. We introduce two refinements to the
kernel dipole tree of \citet{maung2025}, recomputing the pure-or-mixed dipole
classification and the hinge-loss margin locally at each node rather than fixing them
at the root. We formulate the two-stage residual procedure and establish its
theoretical basis. Under a multiplicative spatial frailty and an exogeneity
condition on the conditional distribution of the frailty given the covariates, the
residuals are conditionally free of the clinical covariates and stochastically
ordered by the frailty (Theorem~\ref{thm:residuals}). Under the weaker condition
that only the mean frailty is balanced, an expansion of the frailty Laplace
exponent identifies the leading covariate-dependent remainder
(Proposition~\ref{prop:expansion}).
Both results assume conditions on the data-generating process and consistency of the
clinical stage rather than any particular model at either stage, so the construction
generalizes to other survival estimators. We characterize what the method estimates
when the condition fails, showing that the resulting contrast is well defined but that
the spatial and clinical contributions cease to be separately identifiable in the
strongly confounded regime, a limitation of the estimand shared by smoothing-based
methods rather than a defect of the procedure. Finally, we apply the method to the
LeukSurv data of \citet{henderson2002} alongside a \texttt{spBayesSurv} benchmark, and
report simulations in which the true spatial partition is known, including a sweep
through graded violations of mean balance that locates the identifiability limit
empirically.

The remainder of the paper is organized as follows. Section~\ref{sec:background}
reviews survival trees, dipole splitting, and spatial frailty models.
Section~\ref{sec:method} summarizes the kernel dipole splitting criterion and presents
our two refinements. Section~\ref{sec:threestage} develops the two-stage procedure and
its theory. Section~\ref{sec:leuksurv} applies the method to the LeukSurv data
alongside the benchmark. Section~\ref{sec:simulations} reports the simulation studies.
Section~\ref{sec:discussion} summarizes our contributions and discusses directions
for future work.

\section{Background}
\label{sec:background}

\subsection{Survival Trees and Dipole Splitting}
\label{sec:trees}

Tree-based methods partition the covariate space into regions of homogeneous risk, offering an interpretable alternative to parametric and semiparametric regression models. Survival trees adapt recursive partitioning to censored outcomes, and a substantial literature has developed splitting rules suited to survival data, including the exponential log-likelihood splits of \citet{ciampi1986}, the relative-risk trees of \citet{leblanc1992}, and the log-rank splitting of \citet{leblanc1993}. More recent work has pursued globally optimal survival trees by mixed-integer optimization \citep{bertsimas2022}, which assumes proportional hazards at the splits. Unlike smooth regression models, trees produce piecewise-constant risk surfaces with sharp boundaries, which is precisely the property we exploit for detecting abrupt spatial structure.

Dipole-based splitting takes a different view of how a split should be chosen. Rather than maximizing a separation statistic between the two child nodes, a dipole criterion is built from pairs of observations: a \emph{dipole} is a pair of cases, and a split is favored when it separates pairs that ought to be separated (mixed dipoles, joining cases with dissimilar outcomes) while keeping together pairs that are similar (pure dipoles). This pairwise formulation was developed for classification and regression trees by \citet{bobrowski1996} and \citet{kretowska2004}, and extended to censored survival outcomes by \citet{kretowska2018}, where the classification of a dipole as pure, mixed, or neither is defined through the survival times and censoring indicators. A separate question is \emph{orientation}, which of a paired set of penalty functions each dipole is assigned. This is not rigorously specified in the earlier work and is defined precisely in \citet{maung2025}, where it is what secures the monotone descent of the splitting algorithm. The dipole criterion can be written as a margin objective, which makes it amenable to support-vector-machine optimization and, as in our prior work \citep{maung2025}, to kernelization. Kernel dipole-splitting survival trees replace the axis-aligned splits of classical trees with margin-based separating surfaces in a reproducing kernel Hilbert space, allowing curved, nonlinear partition boundaries while retaining the discrete, piecewise-constant risk structure of a tree. We summarize the criterion and our orientation procedure in Section~\ref{sec:method}.

\subsection{Spatial Survival Models}
\label{sec:spatial}

Spatial survival analysis incorporates geographic information into time-to-event models, typically through a spatially structured frailty added to the linear predictor. The foundational work of \citet{henderson2002} introduced this approach in the context that motivates our application, modeling spatial variation in the survival of acute myeloid leukemia patients in northwest England, using a multiplicative gamma frailty with a Mat\'ern spatial correlation on a proportional-hazards model. Related developments include the spatial frailty proportional-hazards models of \citet{li2002} and the areal frailty models of \citet{banerjee2003}, the latter using conditionally autoregressive priors for region-level data. The framework remains active, with recent geostatistical extensions accommodating incomplete spatially correlated data \citep{allotey2023}. Across these models the spatial effect enters as a smooth, distance-decaying surface. Nearby locations are assigned similar frailties by construction, so the estimated spatial risk varies continuously over the study region.

For our comparisons we use the \texttt{spBayesSurv} package of \citet{zhou2020spbayessurv}, which provides the current standard implementation of Bayesian spatial survival models in R. The package generalizes the Henderson framework along several axes: it supports proportional-hazards, proportional-odds, and accelerated-failure-time models; it models the baseline with a flexible transformed Bernstein polynomial prior, in place of the profiled Breslow estimator of \citet{henderson2002}; and it accommodates both georeferenced (Gaussian random field) and areal (conditionally autoregressive) spatial structure, with scalable Markov chain Monte Carlo via a full-scale covariance approximation. We adopt as our primary benchmark the proportional-hazards model with a Gaussian random field frailty fit on the exact patient locations, which is the closest modern analogue of the point-level model of \citet{henderson2002}. This choice is deliberate. Even with the additional flexibility \texttt{spBayesSurv} provides over the original Henderson model, its spatial component remains a smooth Gaussian process. The frailty surface cannot represent sharp boundaries, so a smooth frailty model blurs the transition wherever the underlying spatial risk changes abruptly, whether across an administrative border, a catchment edge, or an environmental discontinuity. This smoothness limitation, shared by all correlation-based spatial frailty models, is the gap our discrete tree-based method is designed to fill, and the contrast between a smooth frailty surface and a piecewise-constant tree partition is the organizing theme of our empirical comparisons.

\section{Kernel-Based Dipole Survival Trees}
\label{sec:method}

\subsection{$L_2$-Regularized Dipole Splitting Criterion}
\label{sec:criterion}

We summarize the kernel dipole-splitting survival tree of \citet{maung2025}, which our spatial procedure builds on, and then, in Section~\ref{sec:local}, describe two refinements introduced in the present work.

A node containing observations $(x_j, t_j, \delta_j)$ is split by a separating surface chosen to respect the survival structure of \emph{dipoles}, pairs of observations. A pair $(j,k)$ is \emph{right-comparable} when the earlier of the two times is an observed event, and the comparable pairs are graded by the magnitude of their time difference $|t_j - t_k|$: a \emph{pure} dipole joins two observations with similar survival times, which a split should keep together, while a \emph{mixed} dipole joins two with dissimilar survival times, which a split should separate. Concretely, fixing quantile levels $0 < \zeta_1 < \zeta_2 < 1$ of the distribution of time differences over right-comparable pairs, a pair of two observed events whose time difference falls below the $\zeta_1$ quantile is pure, a right-comparable pair whose time difference reaches or exceeds the $\zeta_2$ quantile is mixed, and every other pair is neither.

For a candidate hyperplane $\mathbf{v} = (w_0, w)$ acting on the augmented covariate $z_j = (1, x_j)$, each observation contributes a hinge loss relative to a margin $\varepsilon$,
\[
\varphi_j^+(\mathbf{v}) = \max\{0, \varepsilon - \mathbf{v} \cdot z_j\}, \qquad \varphi_j^-(\mathbf{v}) = \max\{0, \varepsilon + \mathbf{v} \cdot z_j\},
\]
and the dipole penalties combine these so that pure dipoles are penalized for being separated and mixed dipoles for being kept together. Which of a dipole's pair of penalties applies is fixed by its \emph{orientation} with respect to a reference vector $\mathbf{v}_i$, given by the sign of $\mathbf{v}_i \cdot (z_j + z_k)$ for a pure dipole and of $\mathbf{v}_i \cdot (z_j - z_k)$ for a mixed one. Summing the penalties so oriented over the dipole index sets gives the $\mathbf{v}_i$-oriented dipole splitting function $\Psi_{\mathbf{v}_i}(\mathbf{v})$ of \citet{maung2025}. Adding a ridge term yields the $L_2$-regularized splitting function
\[
\begin{aligned}
\Psi_{\mathbf{v}_i}(\mathbf{v};\kappa) = \tfrac{1}{2}\lVert w \rVert^2 + \kappa \sum_j \Bigl[ &\beta_{j,\mathbf{v}_i}^+ \max\{0, \varepsilon_j - \mathbf{v}\cdot z_j\} \\
&+ \beta_{j,\mathbf{v}_i}^- \max\{0, \varepsilon_j + \mathbf{v}\cdot z_j\} \Bigr],
\end{aligned}
\]
where the coefficients $\beta_{j,\mathbf{v}_i}^{\pm} \ge 0$ aggregate the dipole weights $\alpha_{jk}$ at observation $j$ under the orientation induced by a reference vector $\mathbf{v}_i$, and $\kappa > 0$ controls regularization. Minimizing $\Psi_{\mathbf{v}_i}(\cdot\,;\kappa)$ is the convex quadratic program
\[
\begin{aligned}
\min_{w_0, w, \xi} \ & \tfrac{1}{2}\lVert w \rVert^2 + \kappa \sum_j (\xi_j^+ + \xi_j^-) \\
\text{s.t.} \ \ & \beta_{j,\mathbf{v}_i}^+(\varepsilon_j - w\cdot x_j - w_0) \le \xi_j^+, \\
& \beta_{j,\mathbf{v}_i}^-(\varepsilon_j + w\cdot x_j + w_0) \le \xi_j^-, \\
& \xi_j^+ \ge 0, \quad \xi_j^- \ge 0, \qquad j = 1, \dots, n.
\end{aligned}
\]
This program satisfies strong duality, and its dual depends on the covariates only through inner products $x_j \cdot x_k$. Replacing these with a Mercer kernel, $K(x_j, x_k) = \langle \phi(x_j), \phi(x_k) \rangle$ for a feature map $\phi : \mathbb{R}^p \to \mathcal{V}$ into an inner product space, replaces the hyperplane vector $\mathbf{v}$ by a surface vector $\boldsymbol{\upsilon} = (\omega_0, \omega) \in \mathbb{R} \times \mathcal{V}$ and $\Psi_{\mathbf{v}_i}$ by its feature-expanded counterpart $\widetilde{\Psi}_{\boldsymbol{\upsilon}_i}$, giving nonlinear splitting surfaces $S(\boldsymbol{\upsilon}) = \{x : \omega_0 + \langle \omega, \phi(x)\rangle = 0\}$ in the original covariate space. We use linear kernels (recovering oblique splits) for the clinical stage and, for the spatial stage, a mixture of Gaussian and polynomial kernels,
\[
K_{\text{spatial}}(u, u') = \lambda_g \exp\!\left(-\frac{\lVert u - u'\rVert^2}{2\sigma^2}\right) + \lambda_p (u\cdot u' + c)^d,
\]
so that the spatial splits can capture both smooth localized structure (Gaussian component) and low-order polynomial trends (polynomial component) in the geographic coordinates.

\subsection{Iterative Orientation and Optimization Algorithm}
\label{sec:algorithm}

The orientation of a dipole, as defined in Section~\ref{sec:criterion}, is relative to a reference vector rather than intrinsic to the dipole, so the criterion is optimized by alternating between orientation and minimization, each solution serving as the reference for the next. Given a current vector $\boldsymbol{\upsilon}_i$, the dipoles are oriented with respect to $\boldsymbol{\upsilon}_i$ to form $\widetilde{\Psi}_{\boldsymbol{\upsilon}_i}(\cdot\,;\kappa)$, the quadratic program is solved to obtain an improved vector, and the dipoles are reoriented with respect to that solution. The process repeats until the objective stabilizes. \citet{maung2025} show that this produces a monotonically nonincreasing, nonnegative sequence of objective values and therefore converges. In the kernelized form the reorientation is carried out directly from the dual solution, since the discriminant $\langle \omega, \phi(x)\rangle = \sum_j (\mu_j^+ - \mu_j^-) K(x_j, x)$ is available without explicit access to the feature map.

\subsection{Tree Growing and Pruning}
\label{sec:growing}

Trees are grown by recursively applying the splitting algorithm. A node becomes terminal when a split would leave a child with fewer than a fixed number of observations, or when the node contains only pure dipoles. A nonparametric survival estimate is fit at each leaf from the observations assigned to it: a Kaplan-Meier estimate when a leaf survival curve is needed for prediction, and the corresponding Nelson-Aalen cumulative hazard when the leaf is used to form the cumulative hazard residuals of Section~\ref{sec:threestage}.

After a tree is grown it is pruned by the split-complexity method of \citet{leblanc1993}. Branches are removed to optimize $G_\alpha(T) = G(T) - \alpha|S|$, where $G(T)$ is the sum of log-rank split statistics over the internal nodes, $|S|$ is the number of internal nodes, and $\alpha$ trades fit against tree size. The complexity parameter is selected by bootstrap validation.

\subsection{Local Recalibration of the Criterion}
\label{sec:local}

Two ingredients of the criterion are fixed by quantiles of pairwise distributions over the training sample, and \citet{kretowska2018} and \citet{maung2025} compute both at the root. Since those distributions concentrate as the tree deepens, a root value grows mismatched to the node it is applied to. We recompute both at every non-terminal node.

The first is the pure-or-mixed classification of the dipoles, whose thresholds are the $\zeta_1$ and $\zeta_2$ quantiles of the time-difference distribution. Rather than reuse labels computed once for every pair in the full sample, we rebuild that distribution from the observations in the node and relabel the pairs. The quantile levels are unchanged, the thresholds they pick out are not.

The second is the margin $\varepsilon$, a low quantile of the pairwise distances in the feature space $\mathcal{V}$, obtained from the Gram matrix as $\lVert \phi(x_j) - \phi(x_k) \rVert = \bigl(K(x_j,x_j) + K(x_k,x_k) - 2K(x_j,x_k)\bigr)^{1/2}$. A root margin can exceed the entire local diameter at depth, leaving every hinge in its linear regime and the criterion unable to discriminate among splits. We recompute it locally, floored at a proportion of the root value,
\[
\varepsilon_{\text{node}} = \max\!\Bigl( \mathrm{quantile}_{q}(\text{nonzero local feature-space distances}), \ \alpha\, \varepsilon_{\text{root}} \Bigr),
\]
with $q$ a low quantile level and $\alpha \in [0.1, 0.3]$. Both are fixed rules applied to a local subset rather than tuned parameters.
A scale calibrated at the root is
mismatched at depth in any application, but in the spatial stage this is clearest,
the covariates there being coordinates.

\section{Two-Stage Residual Approach for Spatial Effects}
\label{sec:threestage}

\subsection{The Two-Stage Procedure}

Suppose we have $n$ patients, each with clinical covariates $x_i \in \mathbb{R}^p$, spatial location $s_i \in \mathbb{R}^2$, observed time $t_i$, and censoring indicator $\delta_i$. Our goal is to detect spatial variation in survival after adjusting for clinical effects, without imposing proportional hazards or smoothness assumptions on the spatial component. The method has two modeling stages, a \emph{clinical} stage and a \emph{spatial} stage, linked by a cumulative hazard residual transformation.

\textbf{Clinical stage.} We first fit a kernel survival tree to $\{(x_i, t_i, \delta_i)\}_{i=1}^n$. We use a linear kernel since clinical effects are typically well-approximated by linear combinations of standard covariates such as age, sex, and laboratory values. The tree partitions patients into terminal nodes (leaves), each containing patients with similar clinical risk profiles. At each leaf $\ell$, we fit a Nelson-Aalen cumulative hazard estimate $\hat{\Lambda}_\ell(t) = \sum_{t_j \le t} d_j / n_j$ from the patients assigned to that leaf, where $d_j$ and $n_j$ are the number of events and the number at risk, respectively, at the ordered event times $t_j$ within the leaf.

For each patient $i$ assigned to leaf $\ell(i)$, we then compute the cumulative hazard residual
\begin{equation}
\label{eq:residuals}
e_i = \hat{\Lambda}_{\ell(i)}(t_i),
\end{equation}
the leaf Nelson-Aalen cumulative hazard evaluated at the patient's observed time. We use the Nelson-Aalen estimator rather than the transformation $-\log \hat{S}_{\ell(i)}(t_i)$ of a leaf Kaplan-Meier survival estimate. The two agree to first order when the number at risk is large \citep[Section 4.2]{kleinmoeschberger2003} but the latter diverges when the leaf Kaplan-Meier estimate reaches zero, which occurs when $t_i$ is the largest observed time in a leaf and corresponds to an event. The Nelson-Aalen estimator is finite in this case and is the natural direct estimator of the cumulative hazard, the quantity the residual is intended to capture. On the LeukSurv data of Section~\ref{sec:leuksurv}, the Kaplan-Meier transformation produced divergent residuals for $2$ of $1{,}043$ patients, whereas the Nelson-Aalen residuals are bounded. The two residual definitions have rank correlation $0.9998$ across all patients, so the choice affects only these boundary cases and leaves the spatial analysis otherwise unchanged.
The residual $e_i$ is the cumulative hazard at time $t_i$ as estimated from clinically similar patients. The censoring indicator $\delta_i$ carries over to the residual unchanged. If patient $i$ was censored, the residual $e_i$ is treated as a right-censored observation. This residual transformation is the bridge between the two stages. It transfers the censored survival structure to a clinically adjusted scale without invoking a proportional hazards assumption.

\textbf{Spatial stage.} We then fit a second kernel survival tree to $\{(s_i, e_i, \delta_i)\}_{i=1}^n$, using only the spatial coordinates as covariates and the residuals $e_i$ as the survival outcome. For the spatial tree, we use a mixture of Gaussian and polynomial kernels to capture both arbitrary nonlinear boundaries and structured low-order spatial effects. The resulting tree partitions the spatial domain into regions whose patients have systematically different residual survival, after clinical effects have been removed.

The output is a spatial risk map, each patient receiving a predicted residual survival time from the spatial tree that can be read as a spatial relative risk net of clinical adjustment.

\subsection{Theoretical Justification}
\label{sec:theory}

The two-stage procedure can be motivated by the classical Cox-Snell residual framework \citep{coxsnell1968} together with a multiplicative frailty assumption analogous to that of \citet{henderson2002}.

Suppose the true conditional hazard has the multiplicative form
\begin{equation}
\label{eq:frailty_model}
\lambda(t \mid x, s) = z(s) \cdot \lambda_{\text{clin}}(t \mid x),
\end{equation}
where $z(s) > 0$ is a positive spatial frailty representing the relative risk at location $s$, and $\lambda_{\text{clin}}(t \mid x)$ is the clinical-only conditional hazard, an arbitrary nonnegative function of $t$ and $x$. The only structural assumption is the multiplicative frailty itself: $z(s)$ scales the hazard proportionally across time, while $\lambda_{\text{clin}}$ is left unrestricted. This is the hazard structure standard in spatial frailty models \citep{henderson2002, banerjee2003}, but with both factors nonparametric. Let $\Lambda_{\text{clin}}(t \mid x) = \int_0^t \lambda_{\text{clin}}(u \mid x)\, du$ denote the corresponding clinical-only cumulative hazard. The survival function for a patient with covariates $x$ at location $s$ is then
\begin{equation}
\label{eq:strue}
S_{\text{true}}(t \mid x, s) = \exp\bigl[-z(s)\,\Lambda_{\text{clin}}(t \mid x)\bigr].
\end{equation}

The clinical tree observes covariates but not location, so pooling patients who share covariate value $x$ pools different frailties. What it can identify is therefore \eqref{eq:strue} averaged over $S \mid X = x$, and by the law of total probability that averaging is on the survival-probability scale,
\[
S_{\text{marg}}(t \mid x) = \mathbb{E}_{S \mid X = x}\!\left[\, S_{\text{true}}(t \mid x, S)\,\right] = \mathbb{E}_{S \mid X = x}\!\left[\exp\bigl(-z(S)\,\Lambda_{\text{clin}}(t \mid x)\bigr)\right].
\]
Converting this marginal survival function to a cumulative hazard through $\Lambda = -\log S$ gives the quantity the clinical tree actually targets,
\begin{equation}
\label{eq:marg}
\Lambda_{\text{marg}}(t \mid x) = -\log S_{\text{marg}}(t \mid x) = -\log \mathbb{E}_{S \mid X = x}\!\left[\exp\bigl(-z(S)\,\Lambda_{\text{clin}}(t \mid x)\bigr)\right].
\end{equation}
The expectation is taken inside the logarithm, so $\Lambda_{\text{marg}}$ is not the average frailty times $\Lambda_{\text{clin}}$ but a nonlinear function of it. This multiplicative form, the leading term of \eqref{eq:expansion}, survives only to first order, as Proposition~\ref{prop:expansion} makes precise.

\begin{theorem}
\label{thm:residuals}
Assume the multiplicative frailty model in \eqref{eq:frailty_model}, and suppose the clinical tree consistently estimates the cumulative hazard it targets: for each patient $i$ and each $t \in [0,\tau]$,
\begin{equation}
\label{eq:cons}
\hat{\Lambda}_{\ell(i)}(t) \xrightarrow{\ p\ } \Lambda_{\text{marg}}(t \mid x_i),
\end{equation}
with $\Lambda_{\text{marg}}$ as in \eqref{eq:marg} and $\tau$ the administrative end of follow-up, so that the observed time satisfies $T_i \le \tau$. Suppose further the \emph{exogeneity condition} that the conditional distribution of the spatial frailty $z(S)$ given $X = x$ does not depend on $x$. Then, conditional on the location $s_i$, the limiting distribution of the residual $e_i$ from \eqref{eq:residuals} does not depend on the clinical covariates $x_i$, and is stochastically ordered by the frailty: a patient at a higher-risk location (larger $z(s_i)$) has a stochastically smaller residual than one at a lower-risk location. Censoring is carried over unchanged.
\end{theorem}

\begin{proof}
By the probability integral transform applied to the true conditional cumulative hazard \citep[Section 2.3]{kleinmoeschberger2003}, conditional on $(x_i, s_i)$ the random variable formed from the latent event time $T_i^{*}$,
\[
\Lambda_{\text{true}}(T_i^{*} \mid x_i, s_i) = -\log S_{\text{true}}(T_i^{*} \mid x_i, s_i) = z(s_i)\, \Lambda_{\text{clin}}(T_i^{*} \mid x_i),
\]
with $S_{\text{true}}$ as in \eqref{eq:strue}, follows a unit exponential distribution. Dividing the unit exponential by the positive constant $z(s_i)$,
\begin{equation}
\label{eq:expclin}
\Lambda_{\text{clin}}(T_i^{*} \mid x_i) \sim \text{Exp}\bigl(z(s_i)\bigr),
\end{equation}
a distribution that depends on $s_i$ only and not on $x_i$.

Now define $\psi_x(a) = -\log \mathbb{E}_{S \mid X = x}[\exp(-z(S)\,a)]$, the Laplace exponent of the conditional frailty distribution, so that $\Lambda_{\text{marg}}(t \mid x) = \psi_x(\Lambda_{\text{clin}}(t \mid x))$ by \eqref{eq:marg}. The map $\psi_x$ is strictly increasing, since $\psi_x'(a) = \mathbb{E}[z(S)e^{-z(S)a} \mid x] / \mathbb{E}[e^{-z(S)a} \mid x] > 0$. The residual evaluates $\hat{\Lambda}_{\ell(i)}$ at the \emph{random} time $T_i$, which is not independent of it, so \eqref{eq:cons} does not apply to $e_i$ directly. It upgrades, however. Since $\hat{\Lambda}_{\ell(i)}$ is nondecreasing and $\Lambda_{\text{marg}}(\cdot \mid x_i)$ is continuous and nondecreasing, pointwise convergence on $[0,\tau]$ becomes uniform convergence in probability there, by P\'olya's argument applied along almost surely convergent subsequences \citep[cf.][Lemma 2.11]{vandervaart1998}. Then
\[
\bigl\lvert e_i - \Lambda_{\text{marg}}(T_i \mid x_i) \bigr\rvert \;\le\; \sup_{t \in [0,\tau]} \bigl\lvert \hat{\Lambda}_{\ell(i)}(t) - \Lambda_{\text{marg}}(t \mid x_i) \bigr\rvert \xrightarrow{\ p\ } 0
\]
whatever value $T_i$ takes. Hence $e_i - \Lambda_{\text{marg}}(T_i \mid x_i) \to 0$ in probability, and so $e_i$ and $\Lambda_{\text{marg}}(T_i \mid x_i) = \psi_{x_i}(\Lambda_{\text{clin}}(T_i \mid x_i))$ share a limiting distribution. Under the exogeneity condition the conditional frailty distribution does not depend on $x$, so $\psi_x = \psi$ is a single map, the same for every patient, and
\[
e_i \to \psi\bigl(\Lambda_{\text{clin}}(T_i \mid x_i)\bigr).
\]
For an uncensored patient $T_i = T_i^{*}$, so the argument is the variable in \eqref{eq:expclin}; for a censored patient it is the corresponding right-censored value.

\emph{Conditional independence.} The right-hand side is a fixed transform $\psi$ of an $\text{Exp}(z(s_i))$ variable. Both $\psi$ and the $\text{Exp}(z(s_i))$ distribution are free of $x_i$. Hence, conditional on $s_i$, the limiting distribution of $e_i$ does not depend on $x_i$.

\emph{Monotonicity.} The family $\text{Exp}(z(s_i))$ is stochastically decreasing in its rate $z(s_i)$, and $\psi$ is strictly increasing, so $\psi(\text{Exp}(z(s_i)))$ is stochastically decreasing in $z(s_i)$. The censoring indicator transfers because $t \mapsto \hat{\Lambda}_{\ell(i)}(t)$ is monotone increasing, so a right-censored time maps to a right-censored residual.
\end{proof}

\begin{remark}
\label{rem:consistency}
The hypothesis \eqref{eq:cons} is not automatic. Under the stopping rule of
Section~\ref{sec:growing}, which makes a node terminal when a split would leave a
child below a fixed size, the tree acquires more leaves of roughly constant size as
$n$ grows, so leaf-level estimates need not converge. Consistency presumes instead
a regime in which the partition stabilizes and leaf sample sizes increase with $n$.
We assume it rather than establish it, and return to it in
Section~\ref{sec:discussion}.
\end{remark}

Theorem~\ref{thm:residuals} makes no claim about the distributional form of the residuals. The following proposition quantifies the sense in which they are approximately exponential, under the weaker condition that only the \emph{mean} frailty is constant across covariate strata.

\begin{proposition}
\label{prop:expansion}
Assume the model \eqref{eq:frailty_model}, the consistency assumption of Theorem~\ref{thm:residuals}, and the mean-only condition that $\mathbb{E}[z(S) \mid X = x] = \bar{z}$ does not depend on $x$. Then
\begin{equation}
\label{eq:expansion}
\Lambda_{\text{marg}}(t \mid x) = \bar{z}\,\Lambda_{\text{clin}}(t \mid x) - \tfrac{1}{2}\mathrm{Var}\bigl(z(S) \mid X = x\bigr)\,\Lambda_{\text{clin}}(t \mid x)^2 + O\bigl(\Lambda_{\text{clin}}^3\bigr).
\end{equation}
Under the frailty normalization $\mathbb{E}[z(S)] = 1$, the mean-only condition gives $\bar{z} = 1$, so the leading term is exactly $\Lambda_{\text{clin}}(t \mid x)$ and is identical across covariate strata. Consequently, for small clinical cumulative hazard, the residual $e_i$ is exponential with rate $z(s_i)$ to first order. The leading correction is of order $\Lambda_{\text{clin}}^2$, with the covariate-dependent coefficient $\tfrac{1}{2}\mathrm{Var}(z(S) \mid X = x)$.
\end{proposition}

The proof is a Taylor expansion of the Laplace exponent and is deferred to \ref{app:expansion}.

\begin{remark}
\label{rem:estimand}
Theorem~\ref{thm:residuals} delivers the two properties the spatial tree uses,
conditional independence and monotonicity, without asserting exact
exponentiality. Proposition~\ref{prop:expansion} makes precise the sense in which
the residuals are approximately $\text{Exp}(z(s_i))$: the approximation is first
order in $\Lambda_{\text{clin}}$, and the exact exponential is recovered only when
the within-stratum frailty variance vanishes. Exact exponentiality aids interpretation but is not needed for the spatial
stage, which uses only the ordering.
\end{remark}

\begin{remark}
\label{rem:exogeneity}
The exogeneity condition sharpens what the residuals recover but is not required for the procedure to be well defined. Without any condition on the dependence between $S$ and $X$, the clinical tree
still targets $\Lambda_{\text{marg}}(t \mid x)$ of \eqref{eq:marg}, so the
residuals isolate the spatial effect net of clinical adjustment. Exogeneity is precisely the case in which that contrast coincides with the pure
frailty $z(s)$. When it is violated (the clinical
covariates predict the frailty), the component of risk that is simultaneously
clinical and spatial is not separately identifiable, and the residuals attribute it to the spatial term.
This is a property of the estimand rather than a deficiency of the procedure. Smoothing-based and frailty-based spatial survival models likewise cannot
decompose a risk gradient collinear in clinical and spatial coordinates.
\end{remark}

\begin{remark}
\label{rem:general}
Neither Theorem~\ref{thm:residuals} nor Proposition~\ref{prop:expansion} depends
on the two stages being trees. The only hypothesis either result places on the
estimator is that the clinical stage consistently estimates the
covariate-conditional cumulative hazard $\Lambda_{\text{marg}}(t \mid x)$ from
$x$ alone. The expansion \eqref{eq:expansion} is a property of the conditional
frailty distribution and involves no estimator at all. Any first-stage method
returning a monotone estimate of $\Lambda(t \mid x)$ may therefore be
substituted, and any second-stage method that accepts right-censored outcomes.
What makes this possible is the Cox-Snell form of the residual
\eqref{eq:residuals}: because $e_i$ is a monotone transformation of the observed
time, it is itself a right-censored survival outcome. Nor is the construction
specific to space. Here $s$ may be any variable whose effect is to be isolated, with
the exogeneity condition applied to it in place of location. We use kernel
dipole trees at both stages because the clinical stage can then accommodate
non-proportional hazards and the spatial stage can represent sharp boundaries.
These suit the present application rather than being required by the theory.
\end{remark}

The procedure can also be viewed as a survival analogue of regression kriging \citep{hengl2007}, a standard technique in geostatistics for handling spatially correlated data with covariates. In regression kriging, one fits a regression on covariates, computes residuals, and then fits a spatial model on the residuals. Our approach is the survival counterpart. The multiplicative frailty model becomes additive on the log-cumulative-hazard scale, since
\[
\log \Lambda_{\text{true}}(t \mid x, s) = \log z(s) + \log \Lambda_{\text{clin}}(t \mid x),
\]
and removing the clinical effect from the cumulative hazard isolates $\log z(s)$. The key difference from standard regression kriging is that the censoring structure is preserved. Censored observations remain censored in the spatial stage, where the tree's internal Kaplan-Meier estimates handle them.

\subsection{Limitations and Extensions}

Three limitations of the two-stage procedure deserve mention. First, Theorem \ref{thm:residuals} assumes the clinical tree consistently estimates $\Lambda_{\text{clin}}(t \mid x)$. If the clinical tree is misspecified or undersized, residual clinical signal can leak into the spatial stage and contaminate the spatial analysis. In practice, we recommend growing the clinical tree to a moderate depth and using bootstrap pruning to balance bias and variance.

Second, the multiplicative frailty assumption in \eqref{eq:frailty_model} may be violated if the spatial effect changes the shape of the hazard rather than scaling it. The spatial tree is nonparametric, however, so it can still detect systematic spatial differences in the residuals, even though Theorem~\ref{thm:residuals} no longer applies.

Third, the residuals recover the pure spatial frailty only under the exogeneity
condition of Theorem~\ref{thm:residuals}, a case Remark~\ref{rem:exogeneity}
treats in full. Section~\ref{sec:sim_robustness} locates this limit empirically, and
covariate-side adjustment in the strongly correlated regime is left to future
work.

\section{Application to LeukSurv Data}
\label{sec:leuksurv}

\subsection{Data Description}
\label{sec:data}

We apply the method to the LeukSurv dataset, which records survival for $n = 1{,}043$ cases of acute myeloid leukemia in adults, diagnosed between 1982 and 1998 and registered in northwest England. The dataset was assembled and first analyzed by \citet{henderson2002} and is distributed with the \texttt{spBayesSurv} package \citep{zhou2020spbayessurv}. For each patient it records the survival time in days and a right-censoring
indicator. The four clinical covariates are age, sex, white blood cell count at
diagnosis (\textsc{wbc}), and the Townsend deprivation index (\textsc{tpi}) of the
patient's residential area, higher values indicating greater deprivation. Spatial
information is recorded at two resolutions: the administrative district (one of
$24$) and the exact residential coordinates. We use the exact coordinates.

The clinical covariates are the established prognostic factors for this disease. Older age, higher white blood cell count, and greater deprivation are each associated with shorter survival, while the effect of sex is weak. The spatial question, following \citet{henderson2002}, is whether survival varies geographically after these clinical effects are accounted for. The Townsend index is itself spatially clustered, with $22\%$ of its variance lying between the $24$ administrative districts, so the raw ingredient of clinical--spatial confounding is present and the exogeneity condition of Theorem~\ref{thm:residuals} is not guaranteed a priori. In practice, the two-stage procedure removes the deprivation-aligned component effectively. The cumulative-hazard residuals that enter the spatial stage retain almost none of the \textsc{tpi} signal (Pearson $r = -0.063$, Spearman $\rho = -0.068$), and the recovered spatial risk is essentially uncorrelated with \textsc{tpi} (Pearson $r = -0.012$, Spearman $\rho = 0.056$). This is what the exogeneity condition would produce, though the condition concerns the unobserved frailty and cannot be verified directly. The stronger evidence is calibration against known truth. The correlation between \textsc{tpi} and the recovered risk matches the zero-confounding base case of the robustness sweep of Section~\ref{sec:sim_robustness} and falls below every level at which confounding is imposed. The recovered map can accordingly be read as spatial risk net of clinical adjustment.

\subsection{Benchmark Spatial Models}
\label{sec:benchmark}

\paragraph{Henderson et al.\ (2002)} \citet{henderson2002} analyzed these data with Bayesian proportional-hazards
gamma frailty models carrying a Mat\'ern spatial correlation, at two resolutions: a district-level model over the $24$ administrative districts, and a point-level model on individual residential locations. Their substantive conclusion is that survival varies geographically after
clinical adjustment. The posterior median proportion of individual frailty
variance attributable to the spatial component is $0.24$, though the $95\%$
interval runs from $0.04$ to $0.67$. Age and white
blood cell count are clearly prognostic, deprivation is borderline, and sex is not
significant once the spatial term is included. The recovered frailty surface is smooth by
construction, and it is the qualitative target our discrete map is compared
against.

\paragraph{spBayesSurv proportional-hazards Gaussian random field model} As our primary quantitative benchmark we fit the proportional-hazards model with a Gaussian random field spatial frailty from the \texttt{spBayesSurv} package \citep{zhou2020spbayessurv}, using the four clinical covariates and the exact patient coordinates. This is the modern analogue of the Henderson point-level model, a smooth Gaussian-process frailty on the linear predictor, with a flexible transformed Bernstein polynomial baseline in place of the profiled Breslow estimator used there. The model was fit by Markov chain Monte Carlo with a full-scale covariance approximation. Posterior summaries for the regression coefficients and the frailty variance, together with model-fit statistics (LPML, DIC, WAIC), are reported in Table~\ref{tab:spbayes}, and the posterior-mean frailty surface is shown alongside our discrete map in Figure~\ref{fig:comparison}.

\begin{table}[t]
\centering
\small
\setlength{\tabcolsep}{5pt}
\caption{Posterior summaries from the \texttt{spBayesSurv} PH + GRF fit to
LeukSurv: regression coefficients and frailty variance ($n = 1{,}043$; 4{,}000
burn-in, 4{,}000 retained MCMC iterations, 150 full-scale-approximation knots).
Model-fit criteria are given for this fit and for the same model without the
spatial term, which is identical in all other respects.}
\label{tab:spbayes}
\begin{tabular}{lcccc}
\toprule
Parameter & Mean & Median & 95\% CI low & 95\% CI high \\
\midrule
Age              & $0.032$ & $0.032$ & $0.028$ & $0.037$ \\
Sex              & $0.064$ & $0.063$ & $-0.068$ & $0.195$ \\
\textsc{wbc}     & $0.0032$ & $0.0032$ & $0.0022$ & $0.0040$ \\
\textsc{tpi}     & $0.030$ & $0.030$ & $0.010$ & $0.048$ \\
Frailty variance & $0.067$ & $0.059$ & $0.032$ & $0.140$ \\
\midrule
\multicolumn{5}{l}{\emph{Model fit} \hfill LPML \hfill DIC \hfill WAIC \hfill} \\
\multicolumn{5}{l}{PH + GRF \hfill $-5938.2$ \hfill $11870.1$ \hfill $11876.0$ \hfill} \\
\multicolumn{5}{l}{PH, no spatial term \hfill $-5950.6$ \hfill $11898.7$ \hfill $11901.2$ \hfill} \\
\bottomrule
\end{tabular}
\end{table}

The fitted coefficients reproduce the established clinical story: age, white blood cell count, and deprivation are prognostic in the expected direction, and sex is not significant. The benchmark therefore recovers the same clinical effects as \citet{henderson2002}. That the spatial component is doing real work is established by comparison with
the same model fit without
it. Including the Gaussian random field improves LPML
by $12.4$ and lowers DIC and WAIC by $28.6$ and $25.2$ respectively
(Table~\ref{tab:spbayes}). The clinical coefficients are stable across the two
fits, and \textsc{tpi} in particular is essentially unchanged ($0.030$ against
$0.029$), so the spatial term is not absorbing the deprivation effect. The
recovered surface is smooth, risk varying as a continuous gradient with no sharp
boundaries.

\subsection{Our Results}
\label{sec:ourresults}

We applied the two-stage procedure to the LeukSurv data. The clinical stage fit a linear-kernel survival tree to the four clinical covariates (age, sex, \textsc{wbc}, \textsc{tpi}), producing a partition into clinically homogeneous leaves from which the Nelson-Aalen cumulative hazard residuals \eqref{eq:residuals} were computed. The spatial stage fit a survival tree to these residuals using only the geographic coordinates, with the mixture of Gaussian and polynomial kernels of Section~\ref{sec:criterion}, balanced so that the resulting partition is spatially coherent without over-fragmenting. The spatial tree was pruned by the split-complexity method of Section~\ref{sec:growing} at a fixed complexity parameter, applied identically to the spatial-only comparison below so that the two differ only in the clinical adjustment. The clinical tree was left unpruned. The pruned spatial tree has $15$ leaves of $12$ to $223$ patients.

Figure~\ref{fig:comparison} sets the study region and its districts (a) alongside three analyses of the same data on a common centered-log-risk scale: the smooth \texttt{spBayesSurv} posterior frailty surface (b), our discrete two-stage tree risk map (c), and, for contrast, a \emph{spatial-only} tree fit to the raw survival times using the coordinates alone, with no clinical adjustment (d). We first compare the smooth benchmark (b) and the two-stage tree (c), which agree on the spatial structure. Both return a low-risk basin in the west and a low region in the southeast, and both identify three separated elevated regions: a northern lobe, a pocket toward the center-east, and a band across the southwest. This reproduces the geographic risk variation reported by \citet{henderson2002}, whose point-level analysis of the same data resolves two low-frailty clusters, in the west and the southeast, and elevated frailty toward the north, the southwest and the east. Our map recovers each of these, resolving one eastern lobe where theirs resolves two. That both methods independently recover the same separated high-risk regions, not merely a single gradient, is evidence that the discrete method captures the genuine spatial signal rather than an artifact of the partitioning. The two differ on where risk is lowest. The smooth surface places its minimum in the western basin, whereas the tree assigns that basin only a mild reduction and reserves its extreme low values for small pockets nearer the center of the region, one of which adjoins the elevated center-east lobe.

They differ in how they represent that signal. The \texttt{spBayesSurv} surface varies as a continuous gradient, while the tree map is piecewise constant, and where the underlying risk changes abruptly the two diverge. The tree resolves sharp local features that the smooth surface averages into intermediate values, most visibly high- and low-risk pockets lying close together in the center and lower right. The difference is quantitative as well as visual, in that the two-stage map recovers a $5$th-to-$95$th-percentile span of $1.56$ in centered log risk against the smooth surface's $0.56$. In the simulations, where the truth is known, gaps of this size arise because the smooth benchmark recovers roughly $40\%$ of the true frailty contrast while the tree recovers the full contrast, if anything slightly sharpening it (Section~\ref{sec:sim_exogeneity}).

The spatial stage achieves a concordance of $0.57$ on the residual scale, computed in sample over comparable pairs assigned to different leaves. Even a perfectly recovered map cannot achieve high concordance on this scale. Under Theorem~\ref{thm:residuals} the residuals inherit the ordering of $\text{Exp}(z(s))$ variables, so two patients in regions whose frailties differ by a factor $r$ are ordered correctly with probability only $r/(1+r)$. The \texttt{spBayesSurv} frailty variance (Table~\ref{tab:spbayes}) implies a typical ratio of about $1.3$ between two locations, and hence a bound near $0.57$, so the observed value is close to what the benchmark's estimate of the spatial effect permits. Both the smooth and discrete methods agree that the spatial signal, while modest in magnitude, is genuine and geographically organized.

The role of the clinical adjustment is made concrete by the spatial-only map in Figure~\ref{fig:comparison}(d). Fit to the raw survival times without any clinical stage, it diverges from both the smooth benchmark and the two-stage map. It assigns elevated risk to a broad band across the south of the study region containing $553$ of the $1{,}043$ patients. The two-stage map elevates only the western end of that band and returns the remainder to baseline, assigning the band a mean centered log risk of $-0.02$ and placing $69\%$ of its patients below average. The elevation is at least partly clinical in origin. Patients in the band have higher white blood cell counts (median $10.1$ versus $5.2$; Mann-Whitney $p = 7 \times 10^{-6}$) and greater deprivation (median Townsend index $0.50$ versus $-1.33$; $p = 2 \times 10^{-9}$) than patients elsewhere, while being slightly younger (median $63$ versus $67$ years; $p = 0.003$). Age is the strongest prognostic factor in these data and runs in the favorable direction across the band, so the elevated unadjusted risk there cannot be an age effect. It tracks the two covariates that are both prognostic and spatially organized. An analysis that does not adjust for them misreads clinically-driven variation as geographic risk. This is the real-data counterpart of Section~\ref{sec:sim_exogeneity}.

\begin{figure}[htbp]
\centering
\includegraphics[width=\textwidth]{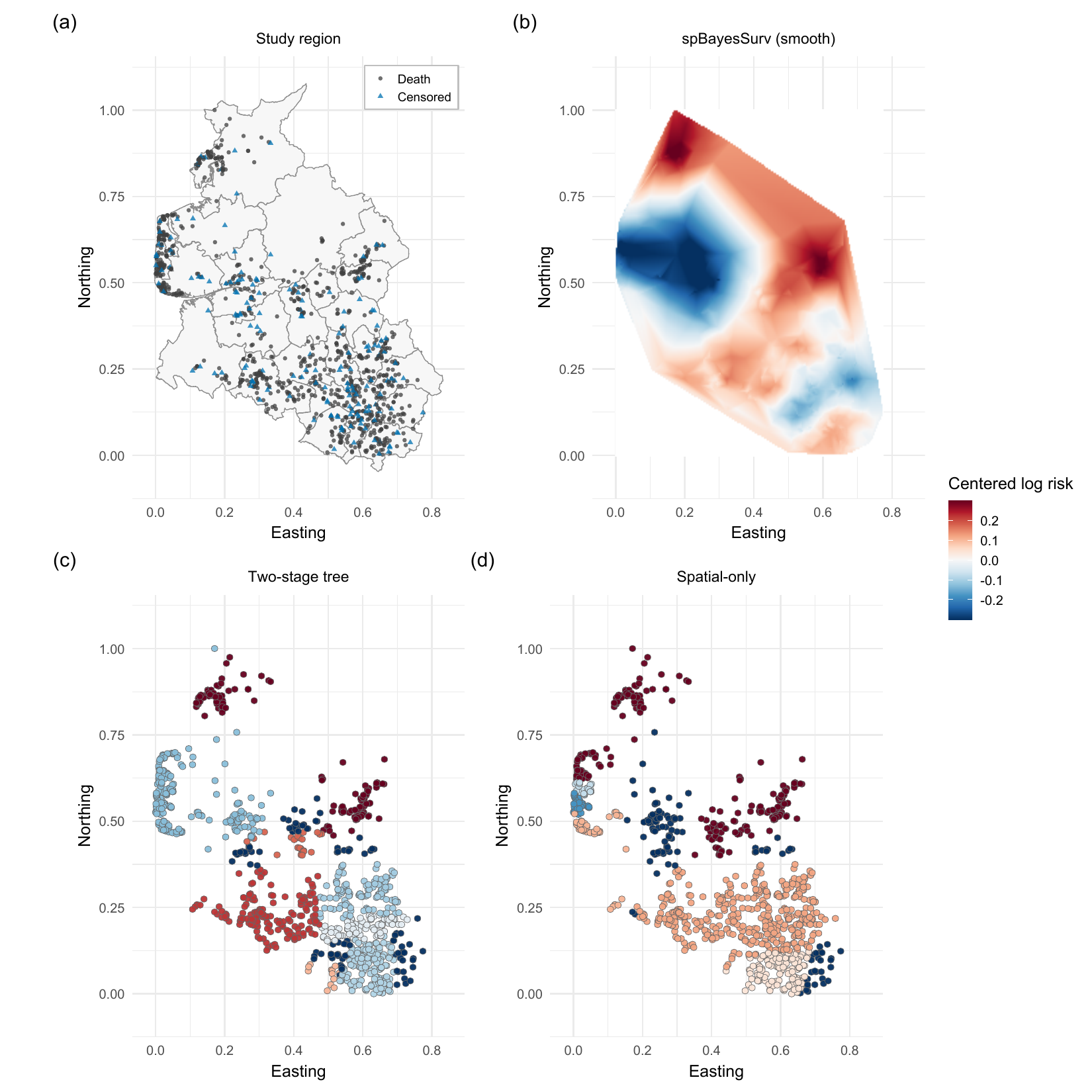}
\caption{LeukSurv: (a) the study region, its $24$ administrative districts and the $1{,}043$ patient locations, deaths as circles and censored observations as triangles; and spatial risk on a common centered-log-risk scale from (b) the smooth posterior frailty surface of the \texttt{spBayesSurv} PH + GRF model, (c) the discrete map from the two-stage kernel dipole tree, and (d) a spatial-only tree fit to the raw survival times with no clinical adjustment.}
\label{fig:comparison}
\end{figure}

\section{Simulation Studies}
\label{sec:simulations}
 
We assess the method on simulated data in which the true spatial structure is known. All designs share a common generative model. The spatial frailty $z(s)$ takes discrete values on a sharp partition of the study region, and the observation coordinates are the real LeukSurv locations, so the spatial support is irregular and realistic rather than a regular grid. Survival times follow the multiplicative frailty model $\lambda(t \mid x, s) = z(s)\,\lambda_{\mathrm{clin}}(t \mid x)$ of Section~\ref{sec:theory}, with an exponential clinical baseline and a crossing-hazard age effect whose log hazard ratio changes sign over follow-up. The age effect therefore violates the proportional-hazards assumption while remaining consistent with Theorem~\ref{thm:residuals}, which places no restriction on the form of $\lambda_{\mathrm{clin}}$. Censoring is set to a moderate rate broadly comparable to that of LeukSurv. We compare three analyses on every design, the spatial-only dipole tree, the two-stage method, and the \texttt{spBayesSurv} PH model with a Gaussian random field frailty, all as in Section~\ref{sec:leuksurv}. The two spatial trees are grown and pruned identically, so they differ only in the clinical adjustment. Because the true zones are known, we report spatial recovery directly. Four
metrics are given for every design: zone-recovery accuracy, the adjusted Rand
index (ARI), the classification accuracy among the half of the observations
farthest from the true boundary (far-boundary accuracy), and the standardized mean squared
error (MSE) between the recovered risk and the true frailty, both centered and
scaled to unit variance. We also report the span
ratio, the $5$th-to-$95$th-percentile span of the recovered risk as a fraction of
the same span of the true frailty. One is exact recovery of the contrast, below
one is shrinkage, and above one is the sharpening that hard leaf assignments
produce. The area under the ROC curve (AUC) for the recovered risk against
the true zone is reported for the two-zone designs.
 
\subsection{The Degenerate Case of Full Independence}
\label{sec:sim_degenerate}
The first design starts at the lower end of the exogeneity ladder. The frailty is a single sharp split into a high-risk and a low-risk zone, and all clinical covariates are generated independently of location, so the strong condition $S \perp X$ holds in full and not merely the weaker condition $z(S) \perp X$ of Theorem~\ref{thm:residuals}. The covariates therefore carry no spatial information, a spatial-only analysis is not biased by them, and the clinical stage has no confounding to remove, so the two tree analyses should coincide. Table~\ref{tab:sim_recovery} reports the recovery metrics for this and the designs that follow. The two-stage and spatial-only analyses recover the two zones comparably and nearly perfectly, with accuracy $0.975$ and $0.987$ and neither making any error away from the boundary. The small edge of the spatial-only analysis is the estimation variance introduced by the clinical stage, which here removes nothing because the covariates are spatially uninformative. The two agree exactly on standardized MSE. They differ sharply in amplitude, however. The two-stage method recovers a span ratio of $0.99$, essentially the true contrast, while the spatial-only tree returns $2.21$, its predicted times carrying the clinical variation as well as the frailty. Recovering the partition and recovering the contrast are separate achievements, and the clinical stage secures the second even where it is not needed for the first. The \texttt{spBayesSurv} benchmark trails the two tree methods on the partition, with an adjusted Rand index of $0.549$ against $0.947$ and $0.903$, even though its risk score ranks the zones nearly as well, with an area under the curve of $0.981$ against $0.986$ and $0.997$. That pattern is the signature of a smooth frailty surface that cannot resolve a sharp boundary. The benchmark is further limited here by the crossing-hazard age effect, which a proportional-hazards model cannot represent.

Figure~\ref{fig:sim_degenerate} shows the recovered risk for each method beside the simulated total risk, in which the clinical variation is scattered without spatial pattern. The degenerate case is a correctness check. The clinical stage does no harm where it is not needed. Figure~\ref{fig:tpi_contrast} contrasts \texttt{tpi} here with the designs that follow, in which it carries spatial structure and the clinical stage becomes essential.

\begin{figure}[t]
\centering
\includegraphics[width=0.9\textwidth]{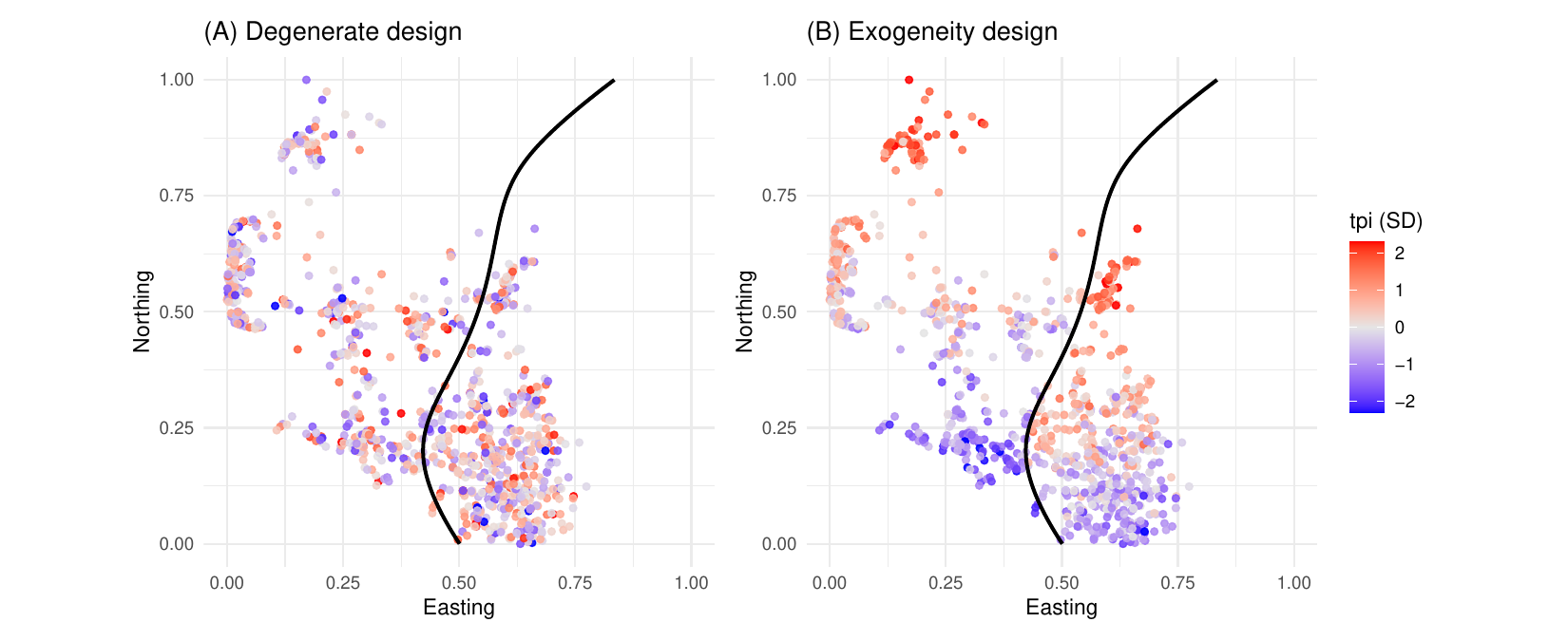}
\caption{The covariate \texttt{tpi} by location, in standard deviation units on a common scale. (A) degenerate design, unstructured; (B) exogeneity design, a north-south gradient that crosses the frailty boundary, overlaid in black, while remaining mean-balanced across the two zones.}
\label{fig:tpi_contrast}
\end{figure}

\begin{figure}[htbp]
\centering
\includegraphics[width=0.9\textwidth]{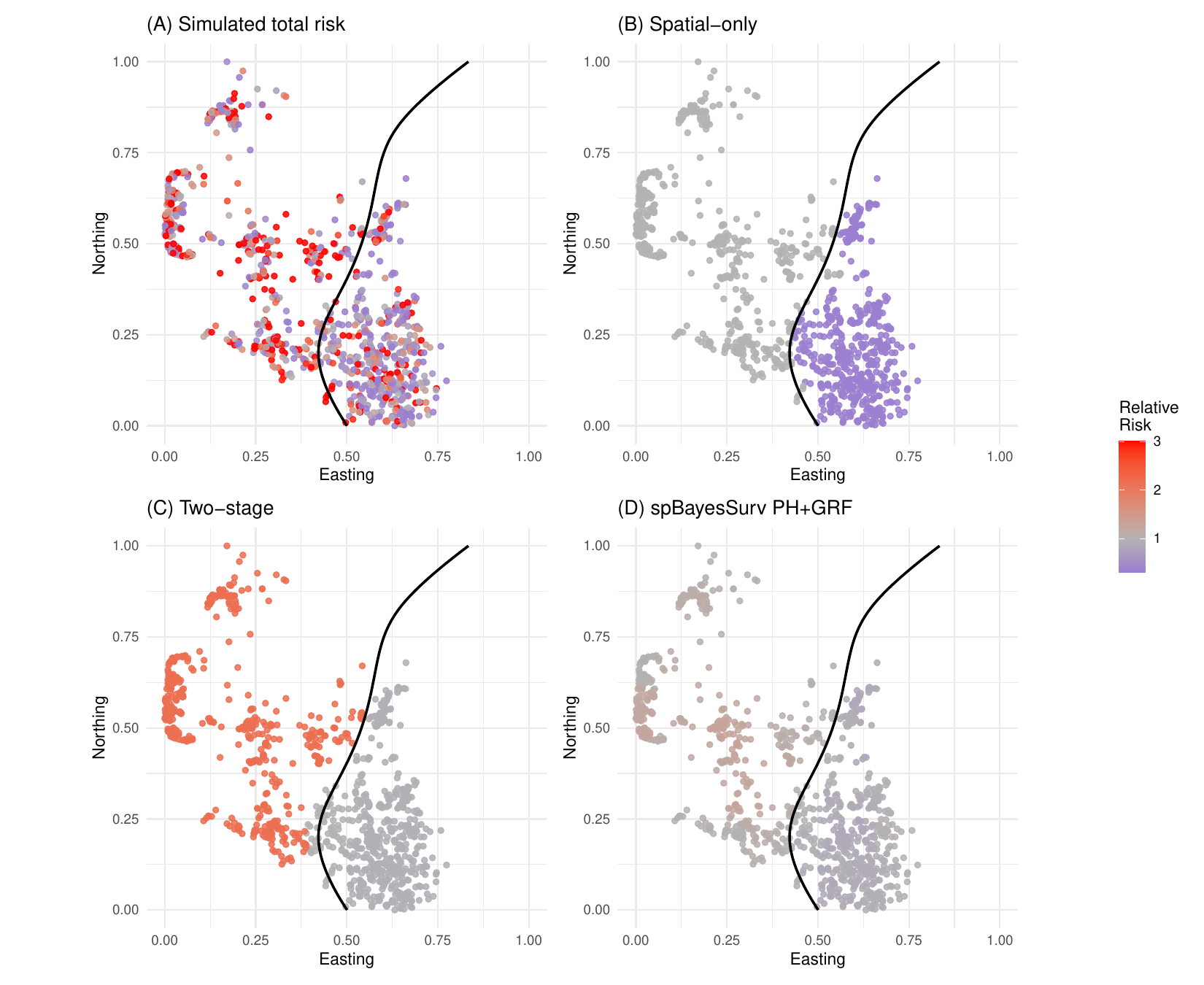}
\caption{Recovered spatial risk in the degenerate design, all panels on a common relative-risk scale. (A) simulated total risk, (B) spatial-only tree, (C) two-stage method, (D) \texttt{spBayesSurv} PH+GRF. The true boundary is overlaid in black, and the benchmark is drawn as per-subject points so that all methods are compared on the same basis.}
\label{fig:sim_degenerate}
\end{figure}

\begin{table}[t]
\centering
\footnotesize
\setlength{\tabcolsep}{4pt}
\caption{Zone-recovery metrics for the four designs. Higher is better for every
column except the standardized MSE, for which lower is better; the span ratio is
relative to the true frailty span, so one is exact recovery of the contrast. AUC
requires a single ROC curve and is reported only where the partition is binary.
Frailty levels are in a uniform ratio of two in the two- and three-zone designs
and $2.5$ in the four-zone design.}
\label{tab:sim_recovery}
\begin{tabular}{llcccccc}
\hline
Design & Method & Acc.\ & Far-bdry.\ & ARI & Std.\ MSE & Span & AUC \\
\hline
Degenerate  & Spatial-only                & \textbf{0.987} & \textbf{1.000} & \textbf{0.947} & \textbf{0.085} & 2.21 & 0.986 \\
$S \perp X$ & Two-stage                   & 0.975 & \textbf{1.000} & 0.903 & \textbf{0.085} & \textbf{0.99} & \textbf{0.997} \\
            & \texttt{spBayesSurv}        & 0.871 & 0.854 & 0.549 & 0.387 & 0.44 & 0.981 \\
\hline
Two-zone    & Spatial-only                & 0.711 & 0.871 & 0.178 & 0.853 & 4.60 & 0.838 \\
            & Two-stage                   & \textbf{0.956} & \textbf{1.000} & \textbf{0.831} & \textbf{0.107} & \textbf{0.79} & \textbf{0.993} \\
            & \texttt{spBayesSurv}        & 0.927 & 0.960 & 0.730 & 0.345 & 0.80 & 0.975 \\
\hline
Three-zone  & Spatial-only                & 0.334 & 0.177 & 0.104 & 1.079 & 2.55 & --- \\
            & Two-stage                   & \textbf{0.784} & \textbf{0.929} & \textbf{0.547} & 0.409 & \textbf{0.79} & --- \\
            & \texttt{spBayesSurv}        & 0.780 & 0.777 & 0.445 & \textbf{0.354} & 0.70 & --- \\
\hline
Four-zone   & Spatial-only                & 0.440 & 0.447 & 0.261 & 0.965 & 2.03 & --- \\
            & Two-stage                   & \textbf{0.882} & \textbf{0.988} & \textbf{0.726} & \textbf{0.368} & \textbf{0.64} & --- \\
            & \texttt{spBayesSurv}        & 0.741 & 0.864 & 0.440 & 0.410 & 0.44 & --- \\
\hline
\end{tabular}
\end{table}

\subsection{Recovery under the Exogeneity Condition}
\label{sec:sim_exogeneity}
The second design places covariates in the regime the method is built for, the regime in which the mean-balance condition of Proposition~\ref{prop:expansion} holds but the clinical stage is not redundant. The clinical covariate \texttt{tpi} is generated as a spatial field, correlated with location, but is centered within each frailty zone. Location and covariates are thus dependent, $S \not\perp X$, while centering delivers mean balance exactly. The full conditional independence of Theorem~\ref{thm:residuals} is approached but not imposed, since the within-zone spread of \texttt{tpi} still differs slightly between zones. Because \texttt{tpi} is spatially structured and prognostic, a spatial-only tree attributes its variation to location and misplaces the recovered boundary, while the two-stage method removes the clinical effect first.

We demonstrate this across increasing spatial complexity, with two, three, and
four zones sharing a single construction. The frailty takes levels in a uniform
ratio between adjacent zones, two in the two- and three-zone designs and $2.5$ in
the four-zone design, so within each design every adjacent pair carries the same
separability and the zone count is a complexity increase rather than a change in
contrast. The boundaries are nested splits of the study region, a parabolic
division of west from east followed, in the three- and four-zone designs, by
further division of the eastern part. Table~\ref{tab:sim_recovery} reports the recovery
metrics and Figure~\ref{fig:sim_maps} the corresponding maps, with the true
boundaries overlaid. The ordering is the same in all three designs. The two-stage method
recovers the partition best, making no error at all away from the boundary in the
two-zone case; the smooth benchmark returns a blurred version, limited both by
the crossing-hazard misspecification and by its inability to represent a sharp
step; and the unadjusted spatial-only analysis is displaced by the clinical
confounding, in each design by an elevated region in the north driven by the
north-south \texttt{tpi} gradient rather than by frailty. In the three-zone design
its accuracy of $0.334$ falls below the $0.50$ a map that resolves neither extreme
would attain, so it does not merely fail to recover the zones but misassigns
them. Absolute recovery falls
as zones are added, reflecting the difficulty of resolving more adjacent zones
rather than any change in method performance.

The mechanism is visible in the residuals themselves, before the spatial stage is fit. Theorem~\ref{thm:residuals} gives residuals ordered by frailty and free of the clinical covariates when its exogeneity condition holds, and Proposition~\ref{prop:expansion} predicts that under mean balance alone a small second-order clinical term remains. Table~\ref{tab:sim_residuals} and Figure~\ref{fig:sim_residuals} confirm both. The residuals are stochastically ordered by zone, with a Mann-Whitney test decisively separating them, and within each zone they separate vertically by frailty while remaining nearly flat against \texttt{tpi}. The removal is substantial but not exact. A small residual dependence survives and is largest in the highest-frailty zone of each design, as Proposition~\ref{prop:expansion} leads one to expect, since the second-order remainder is bounded by the within-stratum frailty variance. The magnitude stays small throughout. In the three-zone design the clinical covariates explain $0.8\%$, $1.6\%$ and $3.0\%$ of the residual variance within the low-, medium- and high-frailty zones, and the correlation between the residuals and northing does not exceed $0.07$ in any zone.

\begin{table}[t]
\centering
\small
\setlength{\tabcolsep}{5pt}
\caption{Residual diagnostics for the two-zone design, computed after the clinical stage and before the spatial stage.}
\label{tab:sim_residuals}
\begin{tabular}{lcc}
\hline
Diagnostic & High zone & Low zone \\
\hline
Mean residual & 0.611 & 1.049 \\
Median residual & 0.467 & 0.835 \\
Correlation with northing & $0.023$ & $0.015$ \\
Clinical gradient $R^2$ (\texttt{wbc}, \texttt{tpi}) & 0.001 & 0.008 \\
\hline
\multicolumn{3}{l}{Mann-Whitney test of zone ordering (high $<$ low): $p = 1.5\times10^{-16}$} \\
\hline
\end{tabular}
\end{table}

\begin{figure}[t]
\centering
\includegraphics[width=\textwidth]{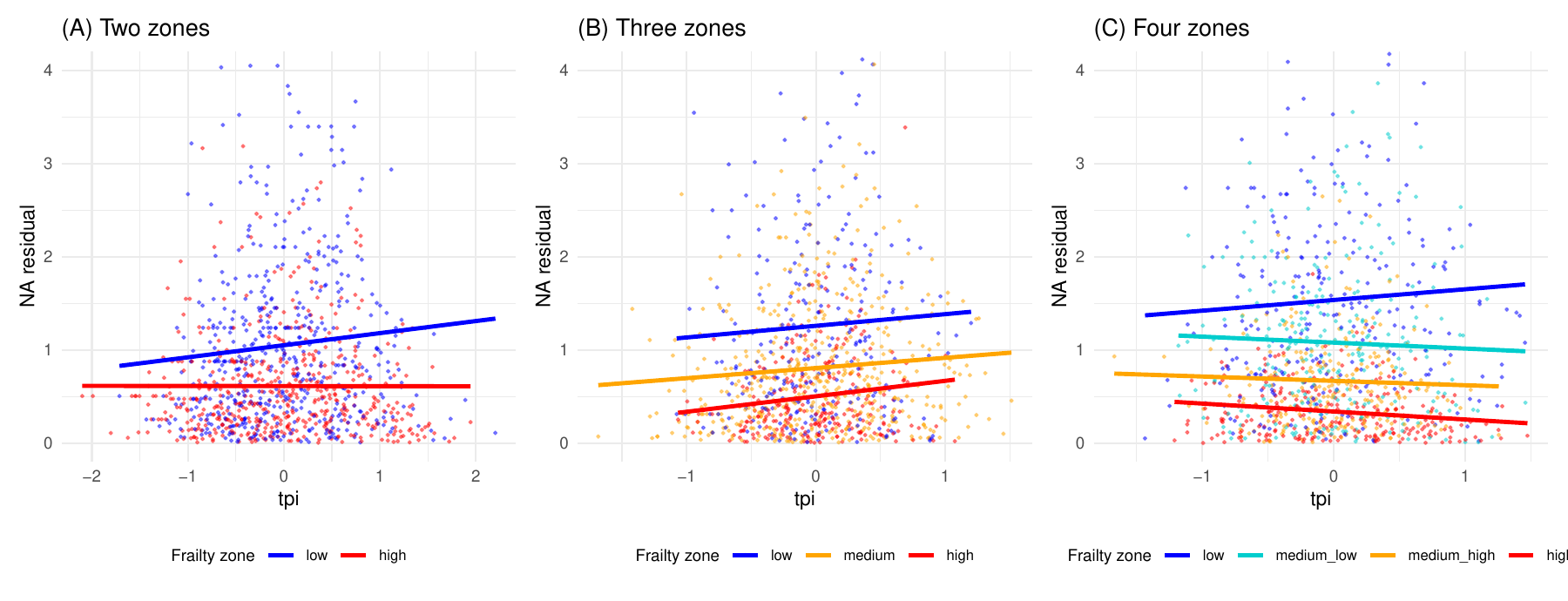}
\caption{Nelson-Aalen residuals against \texttt{tpi} by frailty zone, after the clinical stage, for the two-zone (A), three-zone (B) and four-zone (C) designs on a common residual scale. Fitted lines separate vertically by frailty and are nearly flat against \texttt{tpi} in every zone, the clinical effect having been largely removed; what residual dependence remains is concentrated in the highest-frailty zone, where the second-order term of Proposition~\ref{prop:expansion} is largest.}
\label{fig:sim_residuals}
\end{figure}

\begin{figure}[htbp]
\centering
\includegraphics[width=\textwidth]{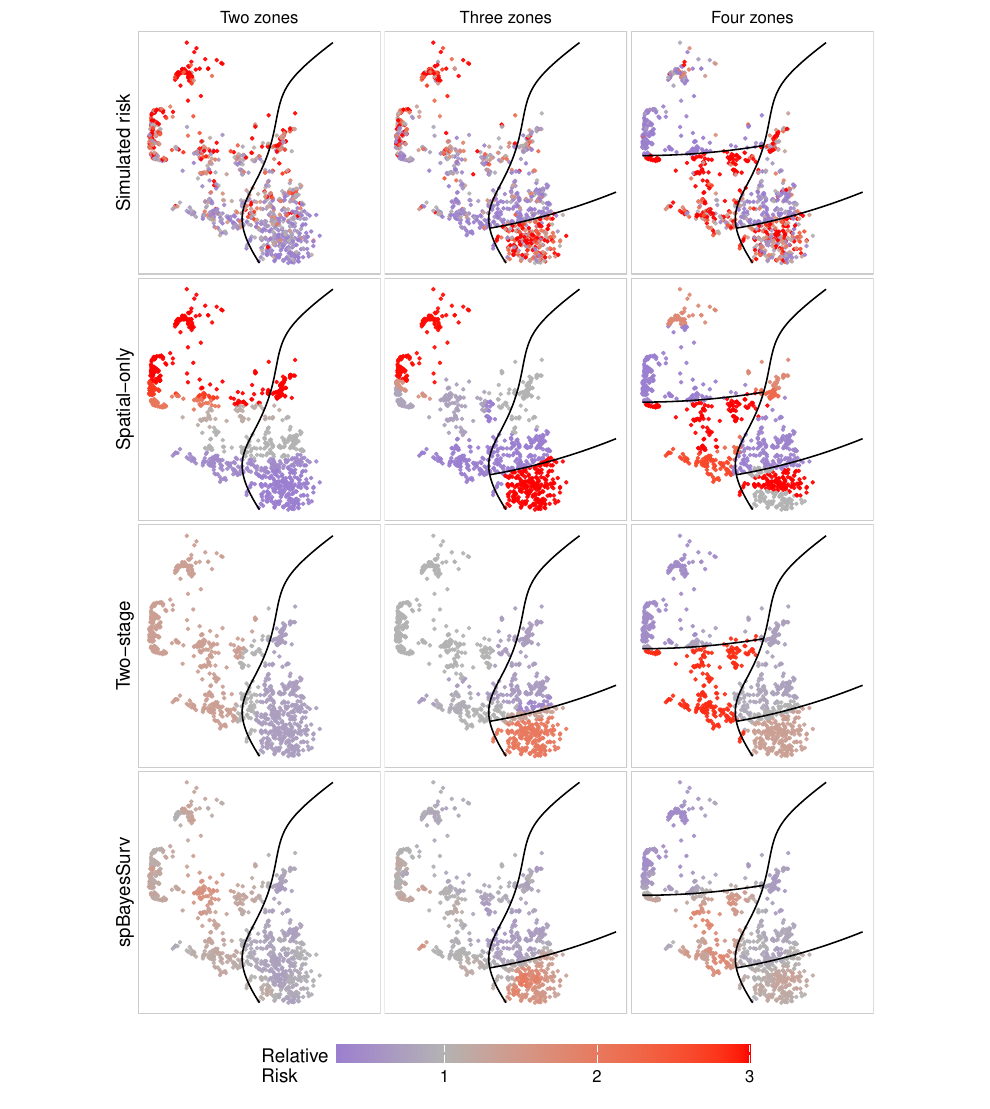}
\caption{Recovered spatial risk under the exogeneity condition, on a common
relative-risk scale. Columns are the two-, three- and four-zone designs; rows are
the simulated total risk and the three analyses. True zone boundaries are
overlaid in black. Axes are the study region as in Figure~\ref{fig:comparison}
and are suppressed for space.}
\label{fig:sim_maps}
\end{figure}

The maps in Figure~\ref{fig:sim_maps} make the difference in kind directly visible. Every panel shares one relative-risk scale, so the muted appearance of the \texttt{spBayesSurv} row reflects genuine shrinkage of the smooth surface rather than a difference in plotting. The two-stage row holds the low-risk regions and resolves the sharp inter-zone steps. The spatial-only row produces comparable contrast but misplaces the boundary, bleeding risk across the true edge. The benchmark row is compressed toward the mean, increasingly so from left to right as zones are added. Against the true frailty span the two-stage tree recovers $0.79$ of the contrast in the two-zone design and $0.64$ in the four-zone, the benchmark $0.80$ and $0.44$. The standardized MSE is blind to this, since both maps are scaled to unit variance before comparison, so the benchmark stays competitive there while losing decisively on the partition metrics.

\subsection{Robustness under Violation of Mean Balance}
\label{sec:sim_robustness}
The final study probes what happens as mean balance fails. Starting from the three-zone design of the previous section, in which the condition holds and the two-stage method recovers the partition well, we introduce a graded violation by coupling the clinical covariate \texttt{tpi} to membership in the high-frailty zone. As the coupling strengthens, high-zone patients receive a systematically elevated \texttt{tpi}, so \texttt{tpi} progressively predicts high-zone membership and the mean-balance condition of Proposition~\ref{prop:expansion} fails. The violation is localized. \texttt{tpi} is shifted upward in the high zone only, leaving the other two untouched, so its effect on recovery can be read cleanly against the confounding strength. The shift also raises the total risk in that zone, since \texttt{tpi} is prognostic, so the strongest coupling makes the confounded zone both harder to attribute and more extreme.

We sweep the coupling from zero, where the exogeneity condition holds, through
the range in which it fails, taking eight levels of the realized correlation
between \texttt{tpi} and high-zone membership, from $0$ to $0.735$, with $30$
independent replicate datasets at each level. All three analyses are fit at every
level, giving $480$ fits in total. At zero confounding the
design is exactly that of Section~\ref{sec:sim_exogeneity}, so the sweep's first
level is thirty replicates of the three-zone design. We record the adjusted Rand
index and accuracy for the whole three-zone partition, a local detection score
for the confounded high zone alone, and the \texttt{tpi} contamination of each
recovered risk map.

Figure~\ref{fig:sim_sweep} shows the result. On the global partition
(panels A and B) the two-stage method holds a plateau through moderate
confounding, with an adjusted Rand index of $0.602$ at zero coupling and no
decline out to a realized correlation of $0.31$, after which it falls steadily
to $0.275$. Accuracy follows the same shape, from $0.767$ to $0.469$. The
\texttt{spBayesSurv} benchmark traces the same trajectory from a lower start,
$0.360$ falling to $0.257$. The spatial-only analysis is flat throughout,
between $0.122$ and $0.154$.

That the two clinically adjusted methods decline together while the unadjusted
one does not is an important observation. It is what the dissolution of the
estimand predicts and a failure of the tree does not. As \texttt{tpi} comes to
determine high-zone membership, the spatial and clinical contributions to that
zone cease to be separately identifiable, and any method that removes the
covariate effect must lose the zone with it. The two-stage method retains its
advantage over the unadjusted analysis on the adjusted Rand index at every
level, from fifteen standard errors across the plateau to six at the strongest
coupling, and over the smooth benchmark from seven standard errors down to
under one. Where the estimand survives, resolving the partition pays; where it
dissolves, the three methods converge.

Panel~(B) marks the majority-zone baseline. The three zones occupy $25\%$,
$50\%$ and $25\%$ of the sample with the majority zone the middle one, so a map
that resolves neither extreme scores $0.50$. Spatial-only accuracy is
significantly \emph{below} that baseline at six of the eight levels and
indistinguishable from it at the other two. The unadjusted map does not merely
fail to recover zone membership, it misassigns it. Two-stage accuracy is far
above the baseline everywhere except the strongest coupling, where it reaches
$0.469$ and is no longer distinguishable from it ($p = 0.13$). The smooth
benchmark stays above the baseline throughout, ending at $0.644$. Its
reluctance to commit to a partition costs it under exogeneity and protects it
once the partition is no longer identifiable.

On the confounded zone alone the ordering reverses. At zero coupling the
two-stage method detects it best, with an area under the curve of $0.946$
against the unadjusted analysis's $0.821$. The curves cross at a realized
correlation of $0.27$, and beyond it the unadjusted analysis leads by a widening
margin, reaching $0.986$ at the strongest coupling while its recovered partition
has an adjusted Rand index of only $0.154$. It identifies the zone by retaining
the very confounding a valid analysis must remove, and it is helped by the
coupling raising that zone's total risk as well as its \texttt{tpi} content. No
method that respects the exogeneity condition can match it there, and the
\texttt{spBayesSurv} curve confirms as much, tracking the two-stage method down
to $0.838$ rather than following the unadjusted analysis up.

Panel~(D) makes the mechanism explicit by reporting how much \texttt{tpi}
signal reaches each recovered risk map. The spatial-only map's correlation with
\texttt{tpi} rises from $0.407$ to $0.752$. At the strongest level its recovered
``spatial'' risk surface is very nearly a map of \texttt{tpi}, which is why its
hotspot detection in panel~(C) scores so highly. The two-stage map stays between $0.055$
and $0.216$, the smooth benchmark between $-0.004$ and $0.248$, and the
residuals entering the spatial stage retain a \texttt{tpi} correlation within
$\pm 0.06$ at every level, so the clinical stage continues to remove the
covariate effect as Theorem~\ref{thm:residuals} requires even where the
theorem's exogeneity condition no longer holds. What degrades is therefore not
the residual transformation but the identifiability of the target.

The practical import for the application is that LeukSurv sits at the clean end
of this range. Its measured \texttt{tpi}/recovered-risk correlation of $0.056$
(Section~\ref{sec:leuksurv}) matches the sweep's zero-confounding base case,
where the same quantity is $0.055$, so the application lies in the regime where
the exogeneity condition holds to good approximation and the two-stage method is
unambiguously preferable. The sweep establishes the margin. The method tolerates
confounding up to a realized correlation of about $0.3$ before its advantage
begins to erode, and the application sits at the base-case level.

\begin{figure}[htbp]
\centering
\includegraphics[width=0.9\textwidth]{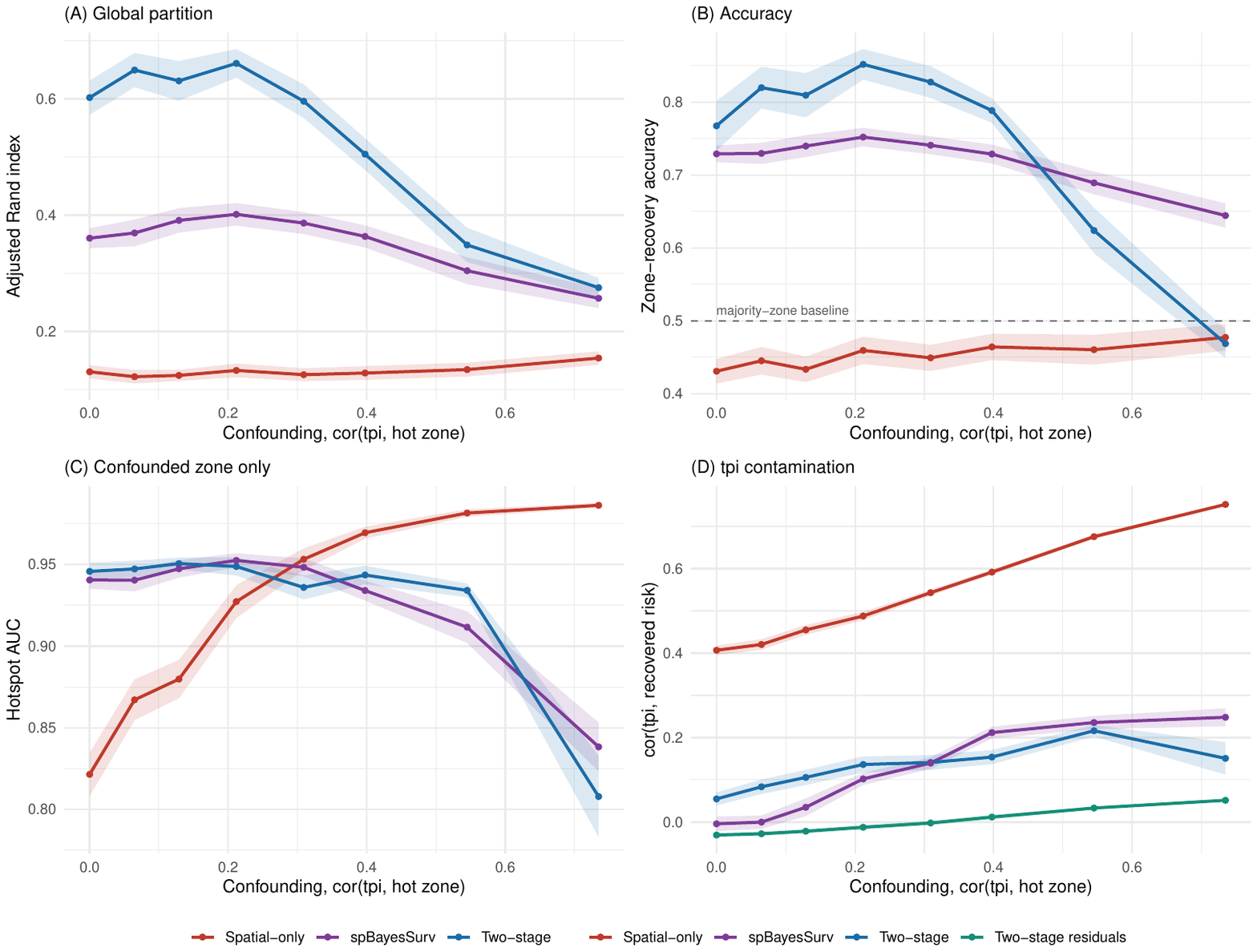}
\caption{Robustness sweep: recovery as mean balance fails, against the realized
strength of \texttt{tpi} coupling to the high-frailty zone. Curves are means over
$30$ replicates per level with $\pm1$ standard-error bands. (A) adjusted Rand
index and (B) accuracy for the full three-zone partition, with the majority-zone
baseline of $0.50$ marked; (C) detection score for the confounded high zone
alone; (D) Spearman correlation between \texttt{tpi} and each method's recovered
risk map, with the \texttt{tpi} correlation of the two-stage residuals
overlaid.}
\label{fig:sim_sweep}
\end{figure}

\section{Discussion}
\label{sec:discussion}

\subsection{Summary of Contributions}

We have developed a nonparametric method for recovering spatial variation in
survival risk after adjustment for clinical covariates, built on kernel
dipole-splitting survival trees and organized around a two-stage residual
transformation. The method makes no proportional-hazards assumption on the
clinical effects, imposes no smoothness on the spatial effect, and returns a
discrete spatial risk map whose boundaries may be curved but whose risk
assignments are hard. Its theoretical basis is
Theorem~\ref{thm:residuals}, which gives the two properties the spatial stage
requires, and Proposition~\ref{prop:expansion}, whose predicted second-order
remainder is visible in the simulated residuals.

Two findings from the empirical work bear emphasis beyond the recovery results
themselves. First, the comparison against an unadjusted spatial-only tree, in both
the LeukSurv application and every simulation design, shows that omitting the
clinical stage does not merely lose precision but produces a qualitatively wrong
map, attributing to location a risk gradient that belongs to the case mix of the
patients living there. Second, the robustness sweep locates the boundary of the
method's validity rather than asserting it. Performance is flat through moderate
confounding and then declines, and the smooth benchmark declines with it while the
unadjusted analysis stays flat and low. What fails is the separate identifiability
of the spatial and clinical contributions, not the estimator.

One caution about evaluation follows from both. A threshold-free ranking statistic
cannot distinguish recovery of a spatial effect from recovery of a covariate that
happens to be spatially coincident with it, as the degenerate design and the
robustness sweep each illustrate. Where the object of inference is a partition,
partition metrics should be reported alongside it.

The method is suited to settings in which risk plausibly changes abruptly and the
location of the change is itself of interest. Care delivered through catchment
areas, screening and treatment protocols set at the level of a trust or a
district, and environmental exposures that stop at a river or a former industrial
boundary all produce risk surfaces a smooth frailty cannot represent. The form of
the output matters as much as its accuracy in such
settings. A discrete map names
regions, and a region is something a health authority can investigate or act on,
whereas a gradient is not. The method requires that the clinical drivers of spatial variation are measured,
since the first stage can only remove what it observes. It also asks something of
how those covariates are arranged in space. Theorem~\ref{thm:residuals} holds when
they are identically distributed across regions of equal frailty;
Proposition~\ref{prop:expansion} asks only that the average frailty be the same in
every covariate stratum. Read geometrically, both are conditions on direction. A
covariate whose gradient runs \emph{along} a risk boundary satisfies them, while
one whose gradient runs \emph{across} it does not. That distinction is often
plausible in practice, because administrative and catchment boundaries are drawn
for historical and organizational reasons rather than to follow socioeconomic
structure, so there is no particular reason for a deprivation gradient to align
with them. Neither condition need hold closely. In the sweep the method loses
nothing until the covariate's correlation with zone membership reaches about
$0.3$.

Relative to the spatial frailty models of \citet{henderson2002}, \citet{li2002}
and \citet{banerjee2003}, and to their modern implementation in
\texttt{spBayesSurv} \citep{zhou2020spbayessurv}, the difference is one of
representation rather than of intent. Both estimate a spatial risk surface after
clinical adjustment, and on LeukSurv they agree about where risk is high and low.
The agreement is what gives us confidence that the discrete map is not an artifact
of partitioning. The same contrast applies to geoadditive and spline-smoothed
proportional-hazards models, in which the spatial term is a penalized smooth rather
than a correlated random
effect. The parameterization differs, but a discontinuity
can still be approximated only by concentrating the smooth's curvature, and how
abruptly it may change is set by the smoothing penalty.

Two limitations relative to the benchmarks should be acknowledged. Our map is a
point estimate with no accompanying interval, where the Bayesian frailty models
return posterior uncertainty for the spatial surface. A single pruned tree is
also a high-variance object, as our own sweep shows, the adjusted Rand index of
the two-stage method having a standard deviation of about $0.16$ across $30$
replicates of an identical design. Both stem from using a single pruned tree
rather than from the two-stage construction, and an ensemble aggregated at the
level of the partition would address them together, as we discuss below.

\subsection{Future Work}

The most direct extension concerns how the spatial and clinical effects are
separated. The present method residualizes the \emph{response}, removing the
clinical contribution from the cumulative hazard before the spatial tree is fit.
An alternative is to residualize the \emph{predictor}, orthogonalizing the
spatial kernel against the clinical covariates in the reproducing kernel Hilbert
space, a kernel analogue of the Frisch-Waugh-Lovell theorem. This acts on the
quadratic form of the dual problem rather than on the outcome, so it composes with
the residual transformation rather than replacing it, and the combination of
response- and predictor-side adjustment may separate the two effects more sharply
than either alone. A third possibility exploits a feature of the dipole criterion
itself, the pair weights $\alpha_{jk}$ of Section~\ref{sec:criterion}, set to unity in
\citet{maung2025}, could instead be chosen to depend on the clinical covariates, concentrating the
criterion on comparisons between clinically similar observations.

Ensemble methods are a natural response to the variance of a single tree, and
prior work has applied them successfully to simpler dipolar trees \citep{kretowska2007, kretowska2014}. The spatial
setting calls for an asymmetric treatment, however, because averaging risk
surfaces across bootstrap replicates would smooth the very boundaries the method
exists to resolve. The clinical tree is not itself reported, its only output
being the residual, so bagging it costs nothing interpretively while stabilizing
the input to the spatial stage, and it strengthens rather than weakens the
consistency hypothesis of Theorem~\ref{thm:residuals}. The spatial tree can
instead be aggregated at the level of the partition rather than the risk, by
cutting a co-occurrence matrix of leaf co-membership across replicates, which
averages away sampling variability while returning a discrete map. Such an
ensemble would also supply what the present method lacks and the Bayesian
benchmarks provide: uncertainty quantification, in the form of bootstrap
confidence in zone membership, and hence a sharp boundary reported together with
its uncertainty.

Several further directions remain. Remark~\ref{rem:consistency} assumes rather than establishes the consistency of
the clinical stage. Whether the split-complexity pruning of
Section~\ref{sec:growing} induces a regime in which it holds is an open problem. Our evaluation is in-sample and
oriented toward spatial recovery, as is standard for spatial frailty models.
Cross-validated prediction on held-out data would situate the method within the
predictive survival literature instead. Finally, the strongly confounded regime invites explicit treatment. Since no
method can separate contributions the data cannot distinguish, a diagnostic that
warns when a dataset has entered that regime would be more useful than one
claiming to recover the decomposition anyway.

\section*{Declaration of competing interest}

The authors declare that they have no known competing financial interests or
personal relationships that could have appeared to influence the work reported in
this paper.

\section*{Data availability}

The LeukSurv data of \citet{henderson2002} and the northwest England district
boundaries are distributed with the \texttt{spBayesSurv} package
\citep{zhou2020spbayessurv}, which also provides the benchmark spatial survival
models fit here. Survival tree models were coded as R6 classes using the
\texttt{data.tree} package \citep{glur2022}, with the kernelized dual problem
optimized by \texttt{osqp} \citep{stellato2020}. The node-splitting survival tree of \citet{maung2025} that both
stages build on, the code for the two-stage procedure and the simulation designs,
and the scripts that reproduce every table and figure are available at
\url{https://github.com/DrewLazar/SpatialSurvTrees}.

\appendix

\section{Proof of Proposition~\ref{prop:expansion}}
\label{app:expansion}

Recall from the proof of Theorem~\ref{thm:residuals} that
\[
\Lambda_{\text{marg}}(t \mid x) = \psi_x\bigl(\Lambda_{\text{clin}}(t \mid x)\bigr), \qquad
\psi_x(a) = -\log \mathbb{E}_{S \mid X = x}\bigl[e^{-z(S)\,a}\bigr],
\]
where $\psi_x$ is the Laplace exponent of the conditional frailty distribution $z(S) \mid X = x$, equal to the negated cumulant generating function evaluated at $-a$. Its Taylor coefficients about the origin are therefore the cumulants of that distribution up to sign, which we verify directly. Let $M(a) = \mathbb{E}[e^{-z(S)\,a} \mid X = x]$, so that $\psi_x = -\log M$ and
\[
\psi_x'(a) = -\frac{M'(a)}{M(a)}, \qquad
\psi_x''(a) = -\frac{M''(a)}{M(a)} + \left(\frac{M'(a)}{M(a)}\right)^2.
\]
At $a = 0$ we have $M(0) = 1$, $M'(0) = -\mathbb{E}[z(S) \mid x]$ and $M''(0) = \mathbb{E}[z(S)^2 \mid x]$, giving
\[
\psi_x(0) = 0, \qquad \psi_x'(0) = \mathbb{E}[z(S) \mid x],
\]
\[
\psi_x''(0) = \mathbb{E}[z(S) \mid x]^2 - \mathbb{E}[z(S)^2 \mid x] = -\mathrm{Var}(z(S) \mid x),
\]
the first two cumulants of the conditional frailty distribution. Assuming $z(S)$ has a finite third moment given $X = x$, a second-order Taylor expansion of $\psi_x$ about $a = 0$ therefore yields
\[
\psi_x(a) = \mathbb{E}[z(S) \mid x]\,a - \tfrac{1}{2}\mathrm{Var}(z(S) \mid x)\,a^2 + O(a^3),
\]
and substituting the mean-only condition $\mathbb{E}[z(S) \mid x] = \bar{z}$ and $a = \Lambda_{\text{clin}}(t \mid x)$ gives \eqref{eq:expansion}. The leading term $\bar{z}\,\Lambda_{\text{clin}}$ is free of $x$ by the mean-only condition; the order-$\Lambda_{\text{clin}}^2$ coefficient $\tfrac{1}{2}\mathrm{Var}(z(S) \mid x)$ is not controlled by that condition and varies with $x$ in general. With $\bar{z} = 1$, the stated consequence for the residuals, that $e_i \sim \text{Exp}(z(s_i))$ to first order, follows by composing three facts, of which the first and third are established in the proof of Theorem~\ref{thm:residuals}: $e_i$ and $\Lambda_{\text{marg}}(T_i \mid x_i)$ share a limiting distribution; $\Lambda_{\text{marg}}(T_i \mid x_i) = \Lambda_{\text{clin}}(T_i \mid x_i) + O(\Lambda_{\text{clin}}^2)$ by \eqref{eq:expansion}; and $\Lambda_{\text{clin}}(T_i \mid x_i) \sim \text{Exp}(z(s_i))$ by \eqref{eq:expclin}. \qed

\bibliographystyle{elsarticle-harv}
\bibliography{references}

\end{document}